\documentclass[10pt,a4paper]{llncs}
\usepackage[margin=2cm]{geometry}
\usepackage{fullpage}
\usepackage[utf8]{inputenc}
\usepackage{amsmath,bbm,amssymb,tabularx,xspace,multirow,mdframed,tikz,makecell,multirow,array,standalone,enumitem}
\newcolumntype{Y}{>{\centering\arraybackslash}X}
\usepackage{subfiles}
\usepackage[ruled,vlined]{algorithm2e}
\usetikzlibrary{positioning}
\tikzset{main node/.style={circle,draw,fill=black,outer sep=0pt,scale=0.35}}
\newcommand{\txtscale}{0.65}
\spnewtheorem{assumption}[theorem]{Assumption}{\bfseries}{\itshape}
\spnewtheorem{observation}{Observation}{\bfseries}{\itshape}

\newcommand{\R}{\mathbb R}

\usepackage{xcolor}
\usepackage[dvipsnames]{xcolor}
\newcommand{\eps}{\varepsilon}

\usepackage{hyperref}
\usepackage{cleveref}
\usepackage[figure]{hypcap}
\usepackage[normalem]{ulem}
\usepackage[bibliography=common]{apxproof}
\newcommand{\Greedy}{\textsc{Greedy}\xspace}

\newcommand{\PrioGreedy}{\textsc{PrioGreedy}$_\text{original}$\xspace}
\newcommand{\PrioGreedyf}{\textsc{PrioGreedy}\xspace}
\newcommand{\PrioCenterf}{\textsc{PrioCenter}\xspace}
\newcommand{\Prioleftf}{\textsc{Prio$_{V_\ell}$}\xspace}
\newcommand{\Alg}{\textsc{Alg}\xspace}
\newcommand{\Opt}{\textsc{Opt}\xspace}

\newcommand{\UBK}{9/4}
\newlist{thmlist}{enumerate}{1}
\setlist[thmlist]{label=\arabic{thmlisti}., ref=\thetheorem.\arabic{thmlisti}}

\renewenvironment{proof}[1][Proof]{\par\noindent\textit{#1.}\ }{\hfill$\square$\par}
\newenvironment{myproof}{%
  \par\noindent\textit{Proof.}\ }{}

\begin{document}
\title{The Buffer Minimization Problem for Scheduling Flow Jobs with Conflicts}
\author{
Niklas Haas\orcidID{0009-0005-7085-8699} \and
Sören Schmitt\orcidID{0000-0002-6117-1343} \and
Rob van Stee\orcidID{0000-0002-3664-0865}
}

\institute{
University of Siegen, Department of Mathematics, Germany\\
\email{niklas.haas@student.uni-siegen.de, \{soeren.schmitt,rob.vanstee\}@uni-siegen.de}
}

\authorrunning{N. Haas, S. Schmitt and R. van Stee} 

\maketitle 
\pagestyle{plain}

\begin{abstract}
    We consider the online \emph{buffer minimization in multiprocessor systems with conflicts problem} (in short, the \emph{buffer minimization problem}) in the recently introduced flow model. In an online fashion, workloads arrive on some of the $n$ processors and are stored in an input buffer. Processors can run and reduce these workloads, but conflicts between pairs of processors restrict simultaneous task execution. Conflicts are represented by a graph, where vertices correspond to processors and edges indicate conflicting pairs. An online algorithm must decide which processors are run at a time; so provide a valid schedule respecting the conflict constraints.
    
    The objective is to minimize the maximal workload observed across all processors during the schedule.
    Unlike the original model, where workloads arrive as discrete blocks at specific time points, the flow model assumes workloads arrive continuously over intervals or not at all.
    We present tight bounds for all graphs with four vertices (except the path, which has been solved previously) and for the families of general complete graphs and complete bipartite graphs. We also recover almost tight bounds for complete $k$-partite graphs.
    
    For the original model, we narrow the gap for the graph consisting of a triangle and an additional edge to a fourth vertex. 
\end{abstract}

\section{Introduction}

In practical multiprocessor settings, it often happens that processors (also called machines) share resources. In such cases, only one of these machines can use the resource at any given time. We model such a situation using a \emph{conflict graph} which specifies for the given set of machines, which are the vertices of the graph, which pairs of machines share resources. This is indicated by having an edge between two machines.

We consider a model with $m$ parallel machines. Load arrives over time, separately to each of the machines. Time is continuous. 
Each machine has a buffer to store unprocessed load.
If a machine is running, it processes load at a rate of 1.
A machine can only run if none of the machines with which it has a conflict (that is, its neighbors in the conflict graph) is running. The goal is to minimize the maximum load that is ever stored in a single buffer. 

We consider the online setting in which the future arriving load is unknown. An online algorithm is compared to the optimal offline solution that can only be achieved if the whole input is known in advance. Note that in the optimal solution, buffers are often also required: consider two neighboring machines that receive load at the same time.

The question now becomes: what size buffers does an online algorithm need, relative to the buffer size in an optimal solution, so that the online algorithm can process any input without overflowing any buffer? This number is known as the competitive ratio, and it depends on the conflict graph and on the number of machines.  Chrobak et al.~\cite{10.1007/3-540-48224-5_70} who introduced this model considered two versions: one where the online algorithm knows the optimal buffer size in advance (the weak competitive ratio) and one where it does not (the strong competitive ratio). In this paper, we will focus on the first measure, which we will usually call simply \emph{competitive ratio}. 

In the original model, load arrives as discrete tasks of various sizes. Höhne and van Stee~\cite{DBLP:journals/tcs/HohneS21} later introduced the so-called \emph{flow model} in which at all times, for any machine, load either arrives at a constant rate of 1 (not as an instantaneous block of load) or not at all. The model is a special case of the original model which maintains its overall complexity. 
An interesting feature of the flow model is that deterministic algorithms appear to be as powerful as randomized algorithms (against an oblivious adversary).

Previously only complete ($k$-partite) graphs and paths have been considered. We go the next natural step and focus on graphs with four vertices. Our goal is to investigate which factors contribute to the problem's difficulty.
In this paper, we consider all graphs with at most four vertices and give tight results for the flow model in all cases. We also give tight or nearly tight results for larger graphs, as well as tight results for almost all graphs with four vertices in the original model. All of our algorithms are short to state, but are carefully constructed to handle various difficult inputs. The proofs involve intricate considerations of what combinations of delays are possible on the various machines. A main building block of the proof consists of invariants for the total delay on pairs or triples of adjacent machines.

\subsection{Related Work}
The \emph{buffer minimization in multiprocessor systems with conflicts} problem was introduced by  Chrobak et al. in 2001 \cite{10.1007/3-540-48224-5_70}. They present a universal algorithm  showing that even the strong competitive ratio on any graph is bounded. Then they consider various families of graphs. For the complete graphs $K_n$ the authors show that the weak and strong competitive ratios are $H_n$ the $n$-th harmonic number, which is achieved by the natural greedy approach. They also show that this approach is almost optimal for complete $k$-partite graphs $K_{n_1,\ldots,n_k}$. If the graph is a tree with diameter at most $\Delta$ they provide an $(1+\Delta/2)$-competitive algorithm. In the final part of their work, the authors give results for the remaining graphs with four vertices. For the path with four vertices they claim that both the weak and strong competitive ratios are at least $2$ and show an upper bound of $5/2$ on both. For the graph $K_3+e$ which you get by connecting an additional vertex to one vertex of $K_3$, they claim a lower bound of $13/6$ for the weak competitive ratio and an upper bound of $5/2$ for the strong competitive ratio. Finally, for $K_4-e$, the graph you get by deleting an edge from $K_4$, they claim that the weak competitive ratio is exactly $2$ and show that the strong competitive ratio is below $5/2.$ 

Höhne and van Stee focus on graphs that are paths \cite{DBLP:journals/tcs/HohneS21}. They close the gap for the path with four vertices by providing a 
$9/4$-competitive algorithm and a matching lower bound. For the path with $5$ vertices they show that the 
competitive ratio is between $16/7$ and $5/2$. For the general paths with at least $6$ vertices they provide a lower bound of $12/5$ and an upper bound linear in the number of vertices of $n/2+1/4.$ 
They also introduce a new model called the flow model as described above. In this model the workload arrives at a constant rate over time in intervals instead of a block arriving at a specific point in time. So the workload function for the machines is continuous in time. This is a cleaner model which avoids the feature of the original model that the competitive ratio is usually achieved on the very last job that arrives, which then often increases the competitive ratio maintained so far by 1. Nevertheless, it is not the case that the competitive ratios for the two models simply differ by 1, and (as can also be seen in this paper) the algorithms for the two models are often quite different on the same graph.

For small paths with up to five vertices Höhne and van Stee present tight results. The paths with two or three vertices can be solved optimally (competitive ratio of 1). For the path with four vertices the competitive ratio is $4/3.$ In contrast to the original model, they give a tight bound of $3/2$ for the path with five vertices. Finally, if the path consists of at least six vertices the competitive ratio is between $3/2$ and $n/2-2/3.$

The concept of conflicts in scheduling problems is well-established, with many variations having been extensively studied. These variations differ in several key aspects like objective function, job model, online/offline and conflict model. While we focus on conflicts between machines, another model that has been widely studied is one where conflicts arise between jobs. Several studies consider models where conflicting jobs cannot be executed concurrently \cite{DBLP:journals/tcs/BakerC96,DBLP:journals/tcs/BodlaenderJ95,DBLP:journals/scheduling/EvenHKR09,DBLP:journals/candie/HongL18,DBLP:conf/soda/IraniL96}, while others address scenarios in which conflicting jobs cannot be scheduled on the same machine \cite{DBLP:conf/mfcs/BodlaenderJ93,DBLP:journals/dam/BodlaenderJW94,DBLP:journals/scheduling/KowalczykL17,DBLP:journals/cor/MallekBB19,9317026,DBLP:journals/itor/MallekB24}. In \cite{DBLP:conf/waoa/BuchemKW22}, conflicts occur between machines, but each job is associated with pre- and post-processing blocking times. The goal in this case is to find a schedule where the blocking times of jobs on conflicting machines do not overlap.

Dósa and Epstein consider a model in which, rather than having individual input buffers for each machine, a universal buffer with a fixed capacity is used to delay job execution \cite{DBLP:journals/jco/DosaE10}. 

\subsection{Results}
We summarize our and known results for the (weak) competitive ratio in the following table.

\begin{table}[ht]
    \caption{Lower and upper bounds on the competitive ratio in the flow model and the original model. Here, $\mu$ is the number of the independent sets in the $k$-partition that consist of a single vertex in a complete $k$-partite graph.}
    \label{tab:results}
    \centering
    \renewcommand*{\arraystretch}{1.2}
\resizebox{\textwidth}{!}
{    \begin{tabularx}{\textwidth}{XXXXX}
    Graph & Flow LB & Flow UB & Original LB & Original UB\\ \hline
    $K_{1,2}$ (Path) & $1$\hfill\cite{DBLP:journals/tcs/HohneS21}  & $1$\hfill\cite{DBLP:journals/tcs/HohneS21} & $2$\hfill\cite{private} & $2$\hfill\cite{10.1007/3-540-48224-5_70}\\ \hline
    $P_4$ (Path) & $4/3$\hfill\cite{DBLP:journals/tcs/HohneS21} & $4/3$\hfill\cite{DBLP:journals/tcs/HohneS21} & $9/4$\hfill\cite{DBLP:journals/tcs/HohneS21} & $9/4$\hfill\cite{DBLP:journals/tcs/HohneS21}\\ \hline
   $K_n$ $(n\geq2)$ & $H_n-\frac{1}{2}$
   & $H_n-\frac{1}{2}$  & $H_n$\hfill\cite{10.1007/3-540-48224-5_70} & $H_n$\hfill\cite{10.1007/3-540-48224-5_70}\\ \hline
   \makecell[l]{$K_{k,n-k}$} & \makecell[l]{$\mu\ge1$: $1$
   \\$\mu=0$: $3/2$} & \makecell[l]{$\mu\ge1$: $1$\\$\mu=0$: $3/2$} & \makecell[l]{$\mu=2$: $3/2$ \\$\mu\le1$: $2$}\hfill\makecell[r]{\cite{10.1007/3-540-48224-5_70}\\\cite{private}} & \makecell[l]{$\mu=2$: $3/2$\\$\mu\le1$: $2$}\hfill\cite{10.1007/3-540-48224-5_70}\\ \hline
  \makecell[l]{complete\\ $k$-partite\\$(k\ge3)$}& \makecell[l]{$\mu\ge1$: \\ $H_{k-1}-\frac{1}{2}+\frac{k-\mu}{k-1}$ \\$\mu=0$: \\ $H_{k-1}+\frac{1}{2}$} & $H_{k-1}+\frac{1}{2}$ &  \makecell[l]{$\mu\ge1$: \\ $H_{k-1}+\frac{k-\mu}{k-1}$\\ $\mu=0$: \\ $H_{k-1}+1$}\hfill\cite{private} & $H_{k-1}+1$\hfill\cite{10.1007/3-540-48224-5_70}\\[5ex] \hline
   $K_3+e$ & $4/3$ & $4/3$ & $13/6$ & $\UBK$\\ \hline 
   $K_4-e$ & $3/2$ & $3/2$ & $2$\hfill\cite{private} & $2$\hfill\cite{private}\\ 
\end{tabularx}
}
\end{table}

\section{Preliminaries}
For simplicity we assume without loss of generality for the rest of the paper that the offline buffer size is $1.$ (Precisely, for any input $I$ the online algorithm receives the scaled input $I'$ such that $\Opt(I')=1$ along with the information that $\Opt(I')=1$.)

We refer to the processors as machines and to the workload as load. The vertices $v_i$ of the conflict graph $G$ represent the machines $m_i$ and an edge indicates a conflict between the adjacent machines, restricting simultaneous task execution. The buffer size is the maximal load on a machine over time. 

For the original model, an input can be described by a list of triples, where the triple $(t,k,s)$ means that a job of size $s$ arrives at time $t$ on machine $k$. However, in many lower bound constructions load arrives only at integer times, and then we will just use a vector $(t;s_1,\dots,s_n)$ to describe the loads arriving at some time $t$, and a list of such vectors to describe a complete input.

For the flow model, an input is a vector of load functions $\R_{\ge0}\to\{0,1\}$, one function for each machine which describes for each non-negative time $t$ whether or not the machine is receiving load at time $t$. In lower bound constructions, similar to in the original model the incoming loads will usually be piecewise constant vectors, changing only at integer times. They will be described as lists of vectors, one vector for each time step.

Our timeline starts at time $t=0.$ Furthermore, as always we only consider finite inputs, meaning there exists a time $T$ after which no machine receives any additional load. We also focus exclusively on connected graphs because the connected components of a graph can be analyzed independently without impacting the overall decision-making process.

\begin{definition}\label{def:CR}
The (strong) competitive ratio of an algorithm \Alg for graph $G$ is defined as $\sup_{I}\frac{\Alg(I)}{\Opt(I)}$, where $I$ is a feasible input. The (strong) competitive ratio of a graph $G$ is defined as $\inf_{\Alg}\sup_{I}\frac{\Alg(I)}{\Opt(I)}$, where \Alg is a feasible algorithm and $I$ a feasible input. The weak competitive ratio is defined analogously, where additionally $\Opt(I)=1.$
\end{definition}
In this paper, we
use competitive ratio to refer to the weak competitive ratio and strong competitive to refer to the competitive ratio where the offline buffer size is unknown in advance.
\begin{definition}
Let $a_i^t$ be the load on machine $m_i$ of an online algorithm \Alg at time $t$ and $z_i^t$ be the load on machine $m_i$ of an optimal offline algorithm \Opt at time $t$. Define the delay of machine $m_i$ at time $t$ as $d_i^t:=a_i^t-z_i^t.$\footnote{We may omit the exponent for time $t$ if we do not need to specify a specific point in time.}

For notational convenience, we denote the load on machine $m_i$ when ignoring the load that arrives exactly at time $t$ by $a_i^{t^-}$. 
\end{definition}

\begin{assumption}
We assume that a buffer size of $1$ suffices if the input was known in advance. Precisely, we assume $z_i\le 1$ at all times and there exists a time $t$ and $i\in\{1,\ldots,n\}$ such that $z_i^t=1$. 
\end{assumption}

\subsection{Flow Model}

In the \emph{flow model} load on machines arrives as a continuous flow over time, in contrast to the original model where load arrives as a single block at a specific point in time. The input is revealed to the online algorithm as follows. At any time $t$, it is aware of the set of machines that are currently receiving load (at speed 1), without knowing for how long these machines will receive load. It can make its schedule starting from time $t$ based on this knowledge; the schedule can change (possibly many times) without the set of machines that receive load changing. The set of machines that are receiving load can change arbitrarily often in any interval.

This restricts the set of allowed inputs; for instance, two connected machines cannot receive load at speed $1/2$ in this model. However, any feasible input can be approximated by an allowed input arbitrarily well by changing back and forth between different sets of machines receiving load. Hence, the power of the adversary is not restricted.
We can approximate any feasible input for this model arbitrarily well by an input for the original model by letting very small blocks of load arrive over time. Comparing this to the flow model described above, it means that the algorithm has a small amount of lookahead (the size of these small blocks), as it is aware of some load that is still to arrive. As the size of the blocks tends to 0, the input tends to an input for our flow model. Conversely, for any input for the original model, a corresponding input for the flow model exists in which each arriving block of load arrives at a rate of 1 instead of instantaneously.

Hence, any upper bound that was proved for the original model still holds; however, we can usually achieve better results in this special case. Conversely, any lower bound for the flow model also holds for the original model. Note that the lower bound for the original model in~\cite{DBLP:journals/tcs/HohneS21} crucially uses that some of the input arrives as a flow, which helps to keep various options open for the adversary.

An algorithm for the flow model must select at each time $t$ an independent set of machines such that each machine either has nonzero load or is currently receiving load. We simplify the setting by allowing algorithms to instead set \emph{speeds} for each machine which can change arbitrarily often. Such an algorithm can be implemented by rapidly changing back and forth between various independent sets of machines. 
The speeds must be such that there is no interval of nonzero length $\ell$ in which any clique of machines processes more than $\ell$ load. Furthermore, for simplicity we allow a machine with zero load to run at nonzero speed, but of course its load does not decrease in this case. Using speeds does not make it easier to design good algorithms (indeed, none of our algorithms set speeds explicitly), but it makes it easier to analyze them.

Initially, all machines are empty, so $a_i^0=0$ for $i\in\{1,\ldots,n\}.$ 

We define the natural \Greedy strategy that we will use as a (sub)routine in many of our presented algorithms. At any time $t$, \Greedy recursively selects a machine with highest load (prioritizing machines receiving load) for running at time $t$ that is not in conflict with any previously selected machine. If there is a tie for the highest load (also in a later step of the recursion), \Greedy will run these machines at equal speeds as discussed above. In the flow model, \Greedy has the following nice property.

\begin{lemma}\label{lem:greedy}
    In the flow model, while \Greedy is running, there is never a single machine which has the unique highest load among all machines.\footnote{This statement does not hold in the original model where jobs arrive as blocks of load.
    }
\end{lemma}
\begin{proof}
    Assume that there exists at some time $t>0$ a machine $m$, that has a strictly higher load than its neighboring machines. We show that this leads to a contradiction. Thus, for any machine, there always exists at least one neighbor with load at least as high as that of the machine itself.
    
    There exists a point in time before $t$, for example at the beginning, at which there is no such machine $m$. Let $t_0\ge 0$ be the supremum over all such times. For the load on $m$ to differ from the loads on its neighbors, a measurable amount of time must have passed between $t_0$ and $t$. Let this time interval be $I:=[t_0,t]$, with $t_0<t$. There must exist a subset $J\subseteq I$ of nonzero measure during which machine $m$ did not run. Else, if $m$ had run throughout the entire interval $I$ (except for a null set), its load would not have increased. Similarly, if the neighbors of $m$ had not run during $I$ (except for a null set), their loads would not have decreased. Hence, machine $m$ would not have more load than its neighbors at time $t$.

    By the intermediate value theorem, there exists a time $t'\in I$ at which $m$ has a higher load than its neighbors, but the difference (to the machine with the closest load) is strictly smaller than at time $t$. Furthermore, we can choose $t'\in J$, since the difference between the load on $m$ and that of its neighbors increases only during $J$.

    The fact that machine $m$ at time $t'$ has load no lower than that of all other machines, while its neighbors have strictly lower loads, but $m$ was not running, contradicts the definition of \Greedy. Since \Greedy prefers machines with higher load, $m$ would have been included in the independent set of machines selected to run.
\end{proof}

\section{Complete Graphs and Complete Partite Graphs}

For families of complete graphs and complete $k$-partite graphs, the \Greedy algorithm performs remarkably well. For complete graphs, \Greedy achieves a competitive ratio of $H_n-\frac{1}{2}$, which cannot be improved even by randomized algorithms. For general complete $k$-partite graphs, we recover almost tight bounds, confirming that \Greedy remains highly effective. However, for complete bipartite graphs, adapting to their specific structure is necessary to obtain tight results.
\begin{definition}\label{def:smoothclique}
If $X$ is a vertex set, then denote $N(X)=\bigcup_{v\in X} N(v)-X,$ where $N(v)$ is the set of neighbors of $v$. We say that $X$ is \emph{smooth} if all vertices in $X$ have the same neighbors outside $X$ that is, $N(v)-X=N(X)$ for $v\in X.$
\end{definition}

\begin{lemma}[Lemma 4.1 \cite{10.1007/3-540-48224-5_70}, Original model]
\label{lem:completecliqueflow}
Suppose $K\subseteq V$ is a smooth clique in G with $N(K)=L.$ Then \Greedy preserves
\begin{equation}\label{eq:completecliqueflow}
\sum_{i\in K} a_i\leq \sum_{i\in K} z_i+\lvert K\rvert\max_{i\in L} a_i.    
\end{equation}
(For $L=\emptyset$ we define $\max_\emptyset a_i=0.$)
\end{lemma}
\begin{theorem}\label{thm:FKnUB}
    \textsc{Greedy} is $(H_n-1/2)$-competitive for $K_n$ for $n\geq2$. 
\end{theorem}

\begin{proof} 
    This follows the proof of the analogous statement in \cite{10.1007/3-540-48224-5_70} for the original model.
    
    We order the machines so that $a_1\geq a_2\geq\ldots\geq a_n.$
    By Lemma \ref{lem:greedy} \Greedy maintains that the load of the two most loaded machines is balanced, so $a_1=a_2$.\footnote{This is the reason why we use $2\leq j$ in the proof.}

    By Lemma~\ref{lem:completecliqueflow} which also applies to the flow model, \Greedy maintains $\sum_K a\leq \sum_K z+\lvert K\rvert\max_L a$ for any choice of $K\subseteq K_n.$ With $a_1\geq a_2\geq \ldots\geq a_n$ and $z_i\leq 1$ for $i\in\{1,\ldots,n\}$, We get that 
    \begin{align*}
         \sum_{i\leq j}a_i&\leq j+j a_{j+1}
    \end{align*}
    for each $j\in\{2,\ldots,n\}$, where we set $a_{n+1}=0.$ Multiplying the $j$-th inequality with $1/(j(j+1))$ for $2\leq j< n$, multiplying the $n$-th inequality with $1/n$ and then adding them all together, we get
    \begin{align*}
         \frac{1}{2}(a_1+a_2)\leq1+\frac{1}{3}+\ldots+\frac{1}{n}=H_n-1/2.
    \end{align*}
    Finally, with $a_1=a_2$ we have $a_1\leq H_n-1/2.$
\end{proof}
\begin{theorem}\label{thm:FKnLB}
    No randomized online algorithm can be better than $(H_n-1/2)$-competitive for $K_n$ with $n\geq2.$
\end{theorem}
\begin{proof} 
    The adversarial sequence consists of $n-1$ phases. At the start of phase $i=1,\dots, n-2$, the adversary sends flow for one time unit on all except the $i-1$ machines of \Alg with the least expected load.  
    After this, nothing arrives for $n-i-1$ time, and then the next phase starts. The last phase $n-1$ consists of sending flow for two time units on the two machines of \Alg with the highest expected load.

    In each phase, \Opt runs all machines that receive load at total speed $1$ except the machine of \Alg that has the lowest expected load after this phase.  After each phase $i$ \Opt has load $0$ on all machines that will receive load in the remaining input and load $1$ on the machines that will not receive any more load. At the end of the construction \Opt has load $1$ on all machines.
    
    We show by induction that at the end of phase $i\leq n-2$ the expected load on the $n-i$ most loaded machines of \Alg is at least $L_i=(n-i)\sum_{j=1}^i \frac{1}{n-j+1}=(n-i)(H_n-H_{n-i}).$ We use $i=0$ as our base case, as at the start of the input we have $L_0=0$ load on all $n$ machines. 
    
     In phase $i+1$ the load on each of the $n-i$ machines with highest expected load at the start of this phase increases by $1$, so at the end of phase $i+1$ the expected load on the (then) $n-(i+1)$ machines with highest expected load is at least $L_{i+1}=\frac{n-(i+1)}{n-i}(L_i+1)$ as the machine among these $n-i$ machines with lowest expected load has expected load at most $\frac{L_i+1}{n-i}.$ 
    
     Using our induction hypothesis, we get
    \begin{align*}
        L_{i+1}&=\frac{n-(i+1)}{n-i}(L_i+1)=\frac{n-(i+1)}{n-i}((n-i)(H_n-H_{n-i})+1)\\
        &=(n-(i+1))(H_n-H_{n-i}+\frac{1}{n-i})=(n-(i+1))(H_n-H_{n-(i+1)}).
    \end{align*}
    
    Then after phase $n-1$ the two machines of \Alg with highest expected load have combined load at least $L_{n-2}+2$, so there is one machine with expected load at least $L_{n-2}/2+1.$

    At the end, there is a machine of \Alg with expected load at least $L_{n-2}/2+1=H_n-1/2.$
\end{proof}

If one independent set consists of a single vertex, a competitive ratio of 1 can be achieved. 
Else, \Greedy achieves a compe\-ti\-tive ratio of $3/2$, which is optimal even for randomized algorithms.

\begin{algorithm}[H]
\label{alg:K1n-1}
\KwData{$K_{1,n-1}=G(V_\ell\cup V_r,E)$ complete bipartite graph with $\vert V_\ell\vert=1$,$V_\ell=\{m_1\}$ and $V_r=\{m_2,\ldots,m_n\}$.}
\eIf{$a_1=0$ or (there exists an $i\in\{2,\ldots,n\}$ such that $a_i=1$ and the machine is receiving load)}{run $m_2,\ldots,m_n$}
{run $m_1$}
  \caption{\Prioleftf}
\end{algorithm}

\begin{theorem}\label{thm:fUBK1n-1}
Algorithm \Prioleftf is $1$-competitive for $K_{1,n-1}.$
\end{theorem}
\begin{proof}
By construction we have $a_i\le1$ for $i\in\{2,\ldots,n\}$ at all times. We show that the invariant $d_1+d_i\le 0$ is maintained for all $i\in\{2,\ldots,n\}$. This is true at the start and whenever $a_j=0$ for all $j\in\{1,\ldots,n\}$. The only time where $d_1+d_i$ could increase if neither $m_1$ nor $m_i$ with $i\in\{2,\ldots,n\}$ run with speed 1. This only happens if $a_i=0$, it does not receive load and there exists a $k\in\{2,\ldots,n\}$ with $a_k=1$ and the machine is receiving load.

Assume that the invariant holds at the start of an interval where such a $k$ exists. With $d_k\ge 0$ we have $d_1\le 0$ at the start of such an interval. This continues to hold as long as there exists a $k\in\{2,\ldots,n\}$ with $a_k=1$ and it receives load. 
During such an interval, while $a_i=0$ and the machine is not receiving load we have $d_i\le 0$ and by $d_1\le 0$ it follows that $d_1+d_i\le 0.$ Note that by the above, if $a_i>0$ or it is receiving load the invariant continues to hold, as the machine runs. Hence, the invariant $d_1+d_i\le 0$ is maintained for all $i\in\{2,\ldots,n\}$ throughout such an interval. Because the invariant is true at the start, it follows that the invariant is maintained at all times.

It remains to show that $a_1\le 1.$ The load $a_1$ could only increase above $1$ if there exists an $i\in\{2,\ldots,n\}$ such that $a_i=1$ and the machine is receiving load, so especially $d_i\ge 0$. By $d_1+d_i\le 0$ it follows $d_1\le 0$, so $a_1\le z_1\le 1.$
\end{proof}

We note that \Prioleftf on $P_3$ is different to the algorithm presented by Höhne and van Stee in \cite{DBLP:conf/faw/HohneS20}, which also achieves $1$-competitiveness. %

\begin{lemma}[Lemma 4.2 \cite{10.1007/3-540-48224-5_70}, Original model]\label{lem:Skpartitedelay}
Suppose that $X\subseteq V$ is a smooth and complete $k$-partite subgraph of $G$ with partite vertex sets $V_1,\ldots, V_k$ and $N(X)=L$. Then \Greedy maintains the invariant
\begin{equation}\label{eq:Skpartitedelay}
    \sum\limits_{i=1}^k\max_{j\in V_i} [a_j-z_j]^+\le \max \{k\max_{i\in L} a_i,k-1\}.
\end{equation} 
\end{lemma}

We next give an upper bound on the performance of \Greedy for bipartite graphs, which we later generalize to $k$-partite graphs in Theorem~\ref{thm:FcompletekpartiteUB}.

\begin{theorem}\label{thm:FcompletebipartiteUB}
   \Greedy is $3/2$-competitive on complete bipartite graphs.
\end{theorem}
\begin{proof}
    Let $G(V_1\cup V_2,E)$ be a complete bipartite graph. By Lemma~\ref{lem:greedy} and the structure of complete biparite graphs, \Greedy maintains $\max_{m_\ell\in V_1} a_\ell=\max_{m_r\in V_2} a_r$. 
    Lemma \ref{lem:Skpartitedelay} also applies to the flow model. With $k=2$ and $N(V_1\cup V_2)=L=\emptyset$ the invariant $\max_{m_\ell\in V_1} d_\ell +\max_{m_r\in V_2} d_r\le 1$ is maintained. This implies $d_\ell+d_v\le1$ for any  pair $(m_\ell,m_v)\in V_1\times V_2.$ So $a_\ell-z_\ell+a_r-z_r\le1\Leftrightarrow a_\ell+a_r\le 3.$ With the above it follows $\max_{m_\ell\in V_1} a_\ell+\max_{m_r\in V_2}\le2\max_{m_\ell\in V_1} a_\ell\le 3,$ so $\max_{m_\ell\in V_1}a_\ell\le 3/2.$ Hence $a_\ell\le 3/2$ for all $m_\ell\in V_1.$ Analogously $a_r\le 3/2$ for all $r\in V_2$ follows.
\end{proof}

 \begin{theorem}\label{thm:FC4LB}
   No randomized online algorithm can be better than $3/2$-compe\-ti\-tive for $C_4$.
\end{theorem}
\begin{proof}
    Let $V_1=\{m_1,m_3\}$ and $V_2=\{m_2,m_4\}$ be the two bipartite vertex sets of $C_4.$ Consider the input which consists of $n$ phases. In each phase send flow for one time unit on the two machines with highest expected load in $V_1$ and $V_2$; these machines are adjacent. Let these machines be $m_\ell\in V_1$ and $m_r\in V_2$. By symmetry we may assume $a_\ell\ge a_r$ after one time unit. In each phase the input continues with flow for one time unit on the other machine in $V_2.$ Precisely, flow arrives on machine $m_r'$ such that $\{m_r'\}= V_2\setminus \{m_r\}.$
    
    We first observe that \Opt can empty all machines at the end of each phase by working on $V_\ell$ for one time unit followed by working on $V_r$ for one time unit.
    
    We let $a_i$ refer to the expected load of machine $m_i$.
    Suppose $a_\ell=x$ and $a_r=y$ at the start of a phase. After the first part of the input we then have $a_\ell+a_r=x+y+1$ after one time unit. By the assumption, this implies $a_\ell\ge (x+y+1)/2.$ Then after one time unit where flow only arrives on $m_r'$, we have $a_\ell+a_r'\ge (x+y+1)/2.$ So as long as $x+y<1$ the overall expected load on a pair of machines that are in conflict increases by $1/2-(x+y)/2$. 
    
    This shows that for any $\eps>0$, no algorithm can avoid having a pair of machines that are in conflict and that have joint expected load more than $1-\eps$ at a time when the optimal machines are empty.
    We then send flow for two time units on both these machines and end with two machines in conflict with combined expected load more than $3-\eps.$ So there is a machine with expected load more than $3/2-\eps/2.$ The lower bound follows.
\end{proof}
Let $\mu$ be the number of the independent sets in the $k$-partition that consist of a single vertex in a complete $k$-partite graph.
\begin{corollary}\label{cor:FLBkpartite}
  No randomized online algorithm can be better than $3/2$-competitive on a complete bipartite graph where $\mu=0$. For $2\ge\mu\ge1$, \Prioleftf is $1$-competitive.
\end{corollary}
\begin{proof}
    This first statement follows immediately from the fact that $C_4$ is a induced subgraph of any complete bipartite graph where $\mu=0$ and Theorem \ref{thm:FC4LB}. If $\mu\ge 1$ we get the complete bipartite graph $K_{1,n-1},$ for which \Prioleftf is $1$-competitive by Theorem~\ref{thm:fUBK1n-1}.
\end{proof}

This shows that for complete bipartite graphs with $\mu=0$, \Greedy is a best possible algorithm and for all other complete bipartite graphs, \Prioleftf is a best possible algorithm.

We now consider general $k$-partite graphs and recover almost tight bounds, where the lower bound deteriorates as $\mu$ increases.
\begin{theorem}\label{thm:FcompletekpartiteUB}
   \Greedy is $(H_{k-1}+\frac{1}{2})$-competitive on complete $k$-partite graphs.
\end{theorem}
\begin{proof} 
Let $G(\bigcup_{i=1}^k V_i,E)$ be a complete $k$-partite graph. Let $A_i=\max_{t\in V_i} a_t$ for all $i=1,\ldots,k$. Reorder the color classes so that $A_1\ge A_2\ge \ldots \ge A_k.$  
Lemma \ref{lem:Skpartitedelay} also applies to the flow model.
Then, (\ref{eq:Skpartitedelay}) implies that for each $j=1,\ldots,k$
\[\sum\limits_{i=1}^j A_i\le \max \{j A_{j+1},j-1\}+j.\]
This holds because when we consider the smooth and complete $j$-partite subgraph $G_{j}=(\bigcup_{i=1}^{j} V_i,E)$ we have $N(\bigcup_{i=1}^{j} V_i)=V\setminus\bigcup_{i=1}^{j} V_i=\bigcup_{i=j+1}^{k}V_{j+1}$ and $A_{j+1}\ge A_{s}$ for $s\ge j+1.$  

Picking the smallest $\ell\le k$ for which $\ell A_{\ell+1}\le \ell -1$ (such a value $\ell$ exists because $A_{k+1}=\max_\emptyset a=0$), we get $\max\{j A_{j+1},j-1\}=j A_{j+1}$ for all $j=1,\ldots,\ell -1$ and $\max\{\ell A_{\ell+1},\ell-1\}=\ell-1.$

For $j=2,\ldots, \ell-1,$ multiply the $j$-th inequality by $1/j-1/(j+1)=1/(j(j+1)),$ multiply the $\ell$-th
 inequality by $1/\ell$, and then add these $\ell-1$ inequalities together. We get
 \[\sum_{j=2}^{\ell-1}\frac{1}{j}\sum_{i=1}^j A_i-\sum_{j=2}^{\ell-1}\frac{1}{j+1}\sum_{i=1}^j A_i+\frac{1}{\ell}\sum_{i=1}^\ell A_i\le \sum_{j=2}^{\ell-1}\frac{1}{j+1}A_{j+1}+H_{\ell-1}+\frac{1}{2},\]
 which after cancellations, yields $\frac{1}{2}(A_1+A_2)\le H_{\ell-1}+\frac{1}{2}\le H_{k-1}+\frac{1}{2}.$ By Lemma~\ref{lem:greedy} \Greedy maintains $A_1=A_2$\footnote{We ignore the first inequality $A_1\le A_2+1$ because the invariant $A_1=A_2$ is a stronger bound.} in the flow model, hence $A_1=\frac{1}{2}(A_1+A_2)\le H_{k-1}+\frac{1}{2}$ and the theorem follows because $A_1$ is the maximal load on the graph.
 \end{proof}
\begin{theorem} 
\label{thm:FcompletekpartiteLB}
    For $k\ge3$, each online algorithm has a competitive ratio of at least $H_{k-1}+1/2$ (resp.,  $H_{k-1}-1/2+(k-\mu)/(k-1)$) for complete $k$-partite graphs with $\mu=0$ (resp.,  $\mu\ge1).$
\end{theorem}
\begin{proof}
This proof is analogous to the proof for the original model provided by Chrobak \cite{private}.

Denote the $k$ independent sets in the
complete $k$-partite graph by $V_1,\ldots, V_k$. Without loss of generality, $\vert V_1\vert, \ldots, \vert V_{k-\mu}\vert>1$ and $\vert V_{k+1-\mu}\vert, \ldots, \vert V_k\vert=1.$

The construction consists of three phases:
\begin{itemize}
    \item \textbf{First phase}: Select $\eps>0$. The goal is to increase total load on $V_1,\ldots,V_{k-\mu}$ to at least $(k-\mu) (1-\eps).$ Select a machine from each set called $m_i$ ($1\le i\le k$). We note that there will always be exactly one selected machine from each independent set. 
    
    Send flow for $1$ time unit on all selected machines. Continue by waiting for $k-2$ time units. We then make the case distinction:
    \begin{enumerate}
        \item If all selected machines $m_i$ of $V_1,\ldots,V_{k-\mu}$ have load at least $1-\eps$, go to second phase. (Here \Opt has total load $1$ at the end.)
        \item There exists at least one machine $m_x$ ($1\le x\le k-\mu$.) with load $a_x<1-\eps.$ Select a machine $m_{x'}\neq m_x$ (works because we consider a non singleton set) from $V_x$ (and replace $m_x$). Send flow on $m_{x'}$ for one time unit.  
        
        We consider how the loads can change:
        
        \emph{Offline}: In the first $k-1$ time units \Opt works on all $V_i$ except $V_x$ and ends with $a_x=1$, all other loads $0=a_i\neq a_x.$ Then it works for one time unit on $V_x$ (so $m_x$ and $m_{x'}$ in parallel) and ends with $a_x=0$ and $a_{x'}$ remains $0$. (Here \Opt has total load 0 at the end; so the round can be restarted.)
        
        \emph{Online:} Assume online has total load $S$ at the start of a round. After the first $k-1$  time units the total load increases to $S+1$ ($k$ load arrives, load can be reduced by $k-1$). Then one unit arrives on $m_{x'}$ (more than $\eps$ time was wasted on $m_x$ in the first $k-1$ time units, because else its load would be at least $1-\eps$), so the total load increases to at least $S+\eps$ after the last time unit.
        
        Precisely, assume that \Alg runs $m_x$ for time $q>\eps$ in the first $k-1$ time units. Then it runs the other machines for time $k-1-q$ and there arrives flow of $k-1$, so the total load increases by at least $q$ on the machines other than $m_x.$ Furthermore, assume that \Alg runs $m_{x'}$ for time $r\le 1$ while load arrives on $m_{x'}$, so for \Alg the load $a_{x'}$ increases by $1-r.$ The remaining time $1-r$ is spent on the machines in $V_1,\ldots,V_k\neq V_x$ and no load arrives there. So at the end the total load increases by at least $(1-r)+(q-(1-r))=q>\eps.$
    \end{enumerate}
    Continue this pattern (if we do not proceed to the second phase by the case distinction).  
    
    If \Alg does not balance the $\eps$ increase over all selected machines in $V_1,\ldots,V_{k-\mu}$, it would eventually break due to the constant increase of $\eps$ for the total load. 
    \item \textbf{Second phase}: Start with \Opt total load $1$ (let it be on machine $m_k$) and \Alg has load at least $1-\eps$ on each selected machine of $V_1,\ldots,V_{k-\mu}$.
    
    Adversary unselects $m_k$, leaving $k-1$ selected machines in $V_1,\ldots,V_{k-1}.$ Offline all these are empty.
    
    If $\mu\ge 1$, online the total load on the selected machines in $V_1,\ldots,V_{k-\mu}$ is still at least $(k-\mu)(1-\eps)$ (as $m_k$ was a in a singleton set).  
    If $\mu=0$, then the online load on the remaining $k-1$ selected machines is at least $(k-1)(1-\eps)$, because we have unselected a machine from a non singleton set.
    
    Move to third phase which only affects the remaining $k-1$ selected machines.
    \item \textbf{Third phase}: Use the lower bound construction in the proof of Theorem \ref{thm:FKnLB} for the complete graph $K_{k-1}$ of the selected machines where a delay of $S$ is already present on the graph.
    
     If $\mu\ge 1$ we have $S=(k-\mu)(1-\eps)$ (because we have no lower bound of load for the selected machines from singleton sets) and if $\mu=0$ we have $S=(k-1)(1-\eps)$. Balancing this delay evenly over all selected machine we get an initial delay of $d_i=(k-\mu)(1-\eps)/(k-1)\to (k-\mu)/(k-1)$ for $\mu\ge 1$ and $d_i=(k-1)(1-\eps)/(k-1)\to 1$ for $\mu=0$ on each machine.
     
     The lower bound construction for complete graphs (Theorem \ref{thm:FKnLB}) is applicable because for $K_{k-1}$ we have $n=k-1\ge 2$, as $k\ge 3$.\footnote{In the flow model if we consider only a single vertex, its load cannot increase if we run it at speed $1$, even if there is delay.} This then yields the bounds $H_{k-1}-1/2+(k-\mu)/(k-1)$ for $\mu\ge1$ and $H_{k-1}-1/2+1=H_{k-1}+1/2$ for $\mu=0.\hfill\qed$
\end{itemize}
\end{proof}

Together with the results for complete bipartite graphs, we get almost tight bounds on general complete $k$-partite graphs.

\section{Graphs with Four Vertices}
We now focus on the graphs with four vertices. There exist six different graphs: $P_4,\linebreak[0] K_4,\linebreak[0] K_{1,3},\linebreak[0] K_{2,2},\linebreak[0] K_4-e$ and $K_3+e.$

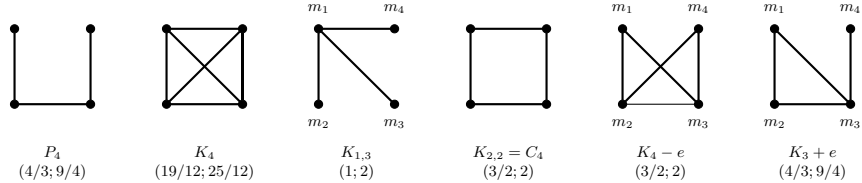
\begin{figure}[ht!]
\centering
\begin{tikzpicture}[on grid, node distance=1cm and 1cm]
	\def\sep{2cm} 
	
	\begin{scope}[shift={(0*\sep,0)}]
		\node[main node, label={[scale=\txtscale]left:}] (P41) at (0,0) {};
		\node[main node, label={[scale=\txtscale]left:}] (P42) [below = of P41] {};
		\node[main node, label={[scale=\txtscale]right:}] (P43) [right = of P42] {};
		\node[main node, label={[scale=\txtscale]right:}] (P44) [above = of P43] {};
		\path[draw,thick] (P41) -- (P42) -- (P43) -- (P44);
		\draw (P42) -- (P43) node [scale=\txtscale, midway, yshift=-1.2cm,align=center] {$P_4$\\$(4/3;9/4)$};
	\end{scope}
	
	\begin{scope}[shift={(1*\sep,0)}]
		\node[main node, label={[scale=\txtscale]left:}] (K41) at (0,0) {};
		\node[main node, label={[scale=\txtscale]left:}] (K42) [below = of K41] {};
		\node[main node, label={[scale=\txtscale]right:}] (K43) [right = of K42] {};
		\node[main node, label={[scale=\txtscale]right:}] (K44) [above = of K43] {};
		\path[draw,thick] (K41) -- (K42) -- (K43) -- (K44) -- (K41) -- (K43) -- (K44) -- (K42);
		\draw (K42) -- (K43) node [scale=\txtscale, midway, yshift=-1.2cm,align=center] {$K_4$\\$(19/12;25/12)$};
	\end{scope}
	
	\begin{scope}[shift={(2*\sep,0)}]
		\node[main node, label={[label distance=1.2mm,scale=\txtscale]above:$m_1$}] (K131) at (0,0) {};
		\node[main node, label={[label distance=1.2mm,scale=\txtscale]below:$m_2$}] (K132) [below = of K131] {};
		\node[main node, label={[label distance=1.2mm,scale=\txtscale]below:$m_3$}] (K133) [right = of K132] {};
		\node[main node, label={[label distance=1.2mm,scale=\txtscale]above:$m_4$}] (K134) [above = of K133] {};
		\path[draw,thick] (K132) -- (K131) -- (K133);
		\path[draw,thick] (K131) -- (K134);
		\draw[opacity=0] (K132) -- (K133)
		node [opacity=1, scale=\txtscale, midway, yshift=-1.2cm,align=center] {$K_{1,3}$\\$(1;2)$};
	\end{scope}
	
	\begin{scope}[shift={(3*\sep,0)}]
		\node[main node, label={[scale=\txtscale]left:}] (C41) at (0,0) {};
		\node[main node, label={[scale=\txtscale]left:}] (C42) [below = of C41] {};
		\node[main node, label={[scale=\txtscale]right:}] (C43) [right = of C42] {};
		\node[main node, label={[scale=\txtscale]right:}] (C44) [above = of C43] {};
		\path[draw,thick] (C41) -- (C42) -- (C43) -- (C44) -- (C41);
		\draw (C42) -- (C43) node [scale=\txtscale, midway, yshift=-1.2cm,align=center] {$K_{2,2}=C_4$\\$(3/2;2)$};
	\end{scope}
	
	\begin{scope}[shift={(4*\sep,0)}]
		\node[main node, label={[label distance=1.2mm,scale=\txtscale]above:$m_1$}] (K4E1) at (0,0) {};
		\node[main node, label={[label distance=1.2mm,scale=\txtscale]below:$m_2$}] (K4E2) [below = of K4E1] {};
		\node[main node, label={[label distance=1.2mm,scale=\txtscale]below:$m_3$}] (K4E3) [right = of K4E2] {};
		\node[main node, label={[label distance=1.2mm,scale=\txtscale]above:$m_4$}] (K4E4) [above = of K4E3] {};
		\path[draw,thick] (K4E1) -- (K4E2) -- (K4E4) -- (K4E3) -- (K4E1);
		\draw (K4E2) -- (K4E3) node [scale=\txtscale, midway, yshift=-1.2cm,align=center] {$K_4-e$\\$(3/2;2)$};
	\end{scope}
	
	\begin{scope}[shift={(5*\sep,0)}]
		\node[main node, label={[label distance=1.2mm,scale=\txtscale]above:$m_1$}] (K3E1) at (0,0) {};
		\node[main node, label={[label distance=1.2mm,scale=\txtscale]below:$m_2$}] (K3E2) [below = of K3E1] {};
		\node[main node, label={[label distance=1.2mm,scale=\txtscale]below:$m_3$}] (K3E3) [right = of K3E2] {};
		\node[main node, label={[label distance=1.2mm,scale=\txtscale]above:$m_4$}] (K3E4) [above = of K3E3] {};    
		\path[draw,thick] (K3E3) -- (K3E1) -- (K3E2) -- (K3E3) -- (K3E4);
		\draw[opacity=0] (K3E2) -- (K3E3)
		node [opacity=1, scale=\txtscale, midway, yshift=-1.2cm,align=center] {$K_3+e$\\$(4/3;\UBK)$};
	\end{scope}
	
\end{tikzpicture}  
\caption{Graphs with four vertices. The numbers $(R_f; R_o)$ represent the competitive ratios in the flow model ($R_f$) and the original model ($R_o$), respectively.  
  }
  \label{fig:graphs}
\end{figure}

We first consider some graphs that are already covered by previous results.
\begin{itemize}
\item The competitive ratio for the path $P_4$ is $4/3$ \cite{DBLP:journals/tcs/HohneS21}.

\item $K_4$ is a complete graph, so \Greedy is $(H_4-1/2=19/12)$-competitive by Theorem \ref{thm:FKnUB} and no algorithm can perform better (Theorem \ref{thm:FKnLB}). 

\item $K_{1,3}$ is a complete bipartite graph with $\mu=1$, so \Prioleftf is optimal (Theorem \ref{thm:fUBK1n-1}).

\item $K_{2,2}=C_4$ is a complete bipartite graph with $\mu=0$, so \Greedy is $3/2$-competitive by Theorem \ref{thm:FcompletebipartiteUB} and no algorithm can perform better (Theorem \ref{thm:FC4LB}).
\end{itemize}

\subsection{$K_4-e$}

Since $K_4-e$ is a complete tripartite graph, \Greedy is at worst $2$-competitive on it. It is easy to verify that the lower bound of 2 
does work for \Greedy (unlike on the complete tripartite graph $K_3$). However, a more sophisticated algorithm can achieve significantly improved competitiveness.

\begin{algorithm}[H]
\label{alg:Fk4-e}
\KwData{$K_4-e$ graph, where $\text{deg}(m_1)=\text{deg}(m_4)=2$ and $\text{deg}(m_2)=\text{deg}(m_3)=3$}
\eIf{$a_1\ge\max\{1,a_2,a_3\}$ or $a_4\ge\max\{1,a_2,a_3\}$ or ($a_2=0$ and $a_3=0$)}{run \Greedy on all machines;}
{run \Greedy on $m_2$ and $m_3$}
  \caption{\PrioCenterf}
\end{algorithm}

\begin{theorem}\label{thm:FUBK4-eUB}
Algorithm \PrioCenterf is $3/2$-competitive on $K_4-e.$ 
\end{theorem}

\begin{proof}
We first show that \PrioCenterf maintains the invariant $d_1+d_2+d_3\leq 0$. It is true at the start. Whenever $a_1>0$ and $m_2$ and $m_3$ are both not running, $m_1$ is running. So during times in which $a_1>0$, $d_1+d_2+d_3$ does not increase.

Suppose $d_1+d_2+d_3\le0$ at the start of an interval in which $a_1>1$. Then since $a_1>0$, this remains true throughout the interval. Moreover, $a_1>1$ implies $d_1>0$. Hence, $d_2+d_3$ remains negative during such an interval. 

Finally, if $a_1=0$, then $d_1\le0$ and $d_2+d_3$ can only increase if $m_4$ is running. This can only happen if $a_2=a_3=0$, in which case the invariants are maintained, or $a_4>1$, so $d_4>0$, which implies $d_2+d_3<0$ as long as we have $d_2+d_3+d_4\le0$ symmetrically to above. 

We conclude that $d_1+d_2+d_3\leq 0$ remains true as long as $d_2+d_3+d_4\le0$ or $a_1>0$ and $d_2+d_3+d_4\leq 0$ remains true as long as $d_1+d_2+d_3\le0$ or $a_4>0$. Both are true at the start, so both remain true (also during times at which $a_1=a_4=0$). We also see from this proof that $d_2+d_3\le0$ always holds. 

By $d_1+d_2+d_3\leq 0$ we have $d_1+d_2\le -d_3=z_i-a_i\le 1.$ By symmetry $d_i+d_j\le 1$ follows for any combination $(i,j)\in\{1,4\}\times\{2,3\}.$

If $a_1>1$, then by the rules of the algorithm $a_1=\max(a_2,a_3)$. Without loss of generality, let $a_2=\max(a_2,a_3)$. Then $d_1>0$ and $d_2>0$. Hence $d_3<0$ and $a_3<1$. Therefore, as long as $a_1>1$, only $m_1$ (and $m_4$) and $m_2$ can be running, and $a_1=a_2$. We also have that $a_2>1$ can only happen if $a_1>1$ (or $a_4>1$), as $d_2+d_3\le0$.  
In this case also $a_1=a_2$ (or $a_4=a_2$).
Since $d_1+d_2\le1$, we conclude that $a_1\le3/2$ at all times.  
By symmetry, the same holds for all machines.
\end{proof}
Sending flow on $m_1$ and $m_3$ for one time unit  
followed by flow on $m_1$ and $m_2$ for two time units shows that this is tight. After one time unit we have $\Alg=(1,0,0,0)$ and $\Opt=(0,0,1,0)$. Then after two more time units the combined load on machines $m_1$ and $m_2$ increases by $2.$ Hence the algorithm balances the load and ends in state $\Alg=(3/2,3/2,0,0)$, while $\Opt=(1,1,1,0).$

 \begin{theorem}\label{thm:Fk4-eLB}
    No randomized online algorithm can be better than $3/2$-com\-petitive for $K_4-e.$
\end{theorem}

\begin{proof}[Proof of \cref{thm:Fk4-eLB}]
The graph $K_4-e$ consists of two triangles (cycles of length 3) that share one edge. We consider an input that consists of phases. 
Each phase starts with load arriving on all vertices of one triangle for one time unit. Then, nothing happens for one time unit. Finally load arrives on all vertices of the other triangle for one time unit, and nothing happens for two time units. At the end of each phase (five time units in total), in the optimal schedule all machines are empty.

Consider some algorithm and let its total expected load on $m_2$ and $m_3$ after the first two time units be $L_1$. Then the total expected load on the second triangle is at least $L_1$ when the phase ends. The next phase starts on this triangle. 

Suppose its competitive ratio is $3/2-\eps$ for some $\eps>0$. We show that the expected load $a_1\le 1-2\eps$, implying $L_1\ge 2\eps$. Suppose for a contradiction that $a_1>1-2\eps$. We continue with sending flow for two time units on $m_1$ and $m_2$ which \Opt could have emptied. Then after two time units we have expected total load more than $1-2\eps+2$ on these two machines, so at least one machine with expected load more than $3/2-\eps,$ which contradicts that the competitive ratio of the algorithm is $3/2-\eps$. 

In phase $n$, the expected total load on the triangle that is used first in this phase is $L_{n-1}+3-2$ after two time units, and again $a_1\le1-2\eps$ at this point, so $L_n\ge L_{n-1}+1-(1-2\eps)=L_{n-1}+2\eps$.
We see that $L_n$ increases linearly with $n$, meaning that the algorithm is not competitive at all, contradicting the assumption that its competitive ratio is better than $3/2$.
\end{proof}

\subsection{$K_3+e$}
In preparation for the flow model, we first consider the original model and narrow the gap between lower and upper bound.

\paragraph{Original Model}
We first observe that the path $P_3$ is an induced subgraph and hence a lower bound of $2$ holds.
The first natural approach is to focus on the machine with highest degree, which we call $m_3$, and run \Greedy on the remaining machines.  The input $((0;1,1,0,0),\linebreak[0](1;0,1,1,0),\linebreak[0](2;0,1,0,0))$  
    shows that this is at best $5/2$-competitive. The algorithm ends in state $(1/2,5/2,0,0)$ and \Opt in $(1,1,1,0).$
This algorithm, however, puts too much emphasis on $m_3.$

\begin{proposition}\label{prop:delay}
  In the original model delays do not change at task arrival, because $a_i^t$ and $z_i^t$ change by the same value.
\end{proposition}

\begin{algorithm}[H]
\label{alg:nFk3+e}
\KwData{$K_3+e$ graph, where $\text{deg}(m_3)=3$ and $\text{deg}(m_4)=1$}
\eIf{$a_3>5/4$}{run $m_3$}
{\eIf{$a_4>5/4$ or $a_1+a_2>1$ or $a_3=0$ 
}
{run $m_4$ and run \Greedy on $m_1$ and $m_2$;}{run $m_3$;}}
  \caption{\PrioGreedy}
\end{algorithm}

\begin{theorem}\label{thm:NFK3+eUB}
In the original model, algorithm \PrioGreedy is $9/4$-com\-petitive on $K_3+e.$ 
\end{theorem}
\begin{proof} 
We first show that \PrioGreedy
maintains the following invariants which imply $a_3\le 9/4$ and $a_4\le 9/4$.
\begin{align*}
    d_3&\leq 5/4\\ 
    d_3+d_4&\leq 5/4\\ 
    d_4&\leq 5/4\\ 
    d_1+d_2+d_3&\leq1\\ 
    d_1+d_2&\leq 1
\end{align*}

The algorithm maintains $d_3\le5/4$. As long as $a_3\le 5/4$ it immediately follows $d_3\le a_3\le 5/4.$ If $a_3>5/4$ at job arrival we have $d_3\le 5/4$ and from this point on $m_3$ runs at maximum speed until $a_3\le 5/4$ again and hence $d_3$ does not increase and $d_3\le 5/4$ holds.

If $a_4=0$, it follows that $d_3+d_4\le d_3\le5/4$. If $a_4>0$ and $a_3=0$, $m_4$ is running and $d_3+d_4\le5/4$ is maintained. Finally, if both machines have load, one of them is running and again $d_3+d_4\le5/4$ is maintained.

If $a_3>5/4$, then $d_3>1/4$ so $d_4<1$. Else, $m_4$ can only not be running (meaning that $d_4$ can increase) if $a_4\le5/4$ (so $d_4\le a_4 \le 5/4)$. It follows that $d_4\le5/4$ is maintained throughout.

Next: $d_1+d_2+d_3\le1$. We run these machines at speed 1 except if $a_1+a_2=0$ and $a_3=0$ (but then immediately $d_1+d_2+d_3\le 0 < 1$) or if $a_1+a_2=0$ and $a_4>5/4$. In the latter case we have $d_1+d_2\leq 0$ since $a_1+a_2=0$ and $d_3<1$ since $d_4>1/4$, so $d_1+d_2+d_3\le1$. 

Consider a time with $d_1+d_2=1$. Then $d_1+d_2+d_3\le 1$ implies $d_3\le 0$ so $a_3 \le 1$ and also $a_1+a_2\ge 1$ holds. Thus condition (A) does not hold but condition (B1) holds and $m_1$ or $m_2$ is running. Hence $d_1+d_2$ does not increase and $d_1+d_2\le 1$ holds. 
\begin{lemma}
\label{lem:1}
   If $a_1+a_2<1$, then $d_3<1$.
\end{lemma}
\begin{proof}
Consider an interval in which $a_3>0$. Before this interval, $d_3\le0$. It is possible to reach the state $d_3\ge1$ only if $a_4\le5/4$ (because otherwise we would have $d_4>1/4$, forcing $d_3<1$). By the rules of the algorithm, this can only happen if $a_1+a_2>1$. However, once we indeed reach the state $d_3\ge1$, then it is no longer possible to reach the state $a_1+a_2<1$, because $d_3\ge1$ implies $d_4\le1/4$. This means $a_4\le5/4$, and obviously $a_3>0$. Therefore the algorithm stops running $m_1$ and $m_2$ as soon as $a_1+a_2=1$. At that point it will start running $m_3$.

We conclude that the state $d_3\ge1$ can only be reached while $a_1+a_2>1$, and in this case $a_1+a_2\ge1$ is maintained as long as $d_3\ge1$. Hence $a_1+a_2<1$ implies $d_3<1$.
\end{proof}
Define $\phi_i:=1-z_i$. This is the amount of load that could arrive instantly on machine $i$, given that the input can be feasibly processed using a buffer size of 1. 
We now show that $a_1+\phi_1\le 9/4$ and $a_2+\phi_2\le 9/4$ and hence with the above $9/4$-competitiveness for \PrioGreedy. This will be done by formulating linear programs.

In the original model it is possible to go from an interval in which $a_2$ is bigger to an interval in which $a_1$ is bigger without having $a_1=a_2$ in between. 

First consider an interval with $a_1=a_2$. Suppose $a_1+\phi_1>9/4$ so $d_1>5/4$. By $d_1+d_2\le 1$ we get $d_2<-1/4$ so $a_2<3/4$. But to get $a_1+\phi_1>9/4$ we must have $a_1>5/4$ which contradicts $a_1=a_2$. Analogously we can show that $a_2+\phi_2\le 9/4$ holds in intervals with $a_1=a_2$.

We will now consider a maximal interval with starting tine $t$ in which one machine always has the bigger load. Without loss of generality, let this machine be $m_2$.
We can assume that $a_1+\phi_1\le 9/4$ and $a_2+\phi_2\le 9/4$ held until time $t$ inclusive because it holds in intervals with $a_1=a_2$.

In the next part we consider different time steps within an interval and use the notation $x_i^{t}$ (with $x\in\{a,z,d,\phi\}$) to refer to the state at time $t.$ We omit the exponent if we make general statements about these variables.

Let $T$ be the first time which maximizes $a_2^T+\phi_2^T$. If $a_2+\phi_2\le 9/4$ holds until time $T$ inclusive we also have $a_1+\phi_1\le 9/4$ in this time because by invariant $d_1+d_2\le 1$ we would have $d_2<-1/4$ so $a_2<-1/4$ which contradicts that we would need to have $a_1>5/4$ and $a_1<a_2$. Hence we only need to show that $a_2+\phi_2\le 9/4$ holds in the interval $(t,T]$. 
In the interval $(t,T]$, $a_2>a_1$ holds and machine $m_1$ is never run, so $d_1$ does not decrease, while $d_2+d_3$ does not increase (note that by $a_2>a_1\geq 0$ the algorithm runs a machine with load). 
We have $a_2^{t^-}\le a_1^{t^-}\le a_1^{t^-}+\phi_1^{t^-}=1+d_1^{t^-}$ and $\phi_2^{t^-}\le1$, so $a_2^{t}+\phi_2^{t}=a_2^{t^-}+\phi_2^{t^-}\le 2+d_1^{t^-}=2+d_1^t$.

\paragraph{Case 1.}
Suppose $a_1^t<1/2$.
The value $a_2+\phi_2=d_2+1$ can only increase during intervals in which \Opt works on $m_2$ and \PrioGreedy works on $m_3$ (type 1 intervals). 
During such intervals, $d_3$ decreases. If \Opt works on $m_1$ instead, $d_3$ still decreases but $d_2$ remains constant.
In intervals in which \Opt and \PrioGreedy work on the same machine, $d_2$ and $d_3$ remain constant. In intervals in which \Opt works on $m_3$ or $m_1$ and \PrioGreedy works on $m_2$ (type 2 intervals), $a_2+\phi_2$ decreases.

Let $x$ be the total length of time of type 1 intervals minus the total length of time of type 2 intervals in the interval $(t,T]$. Then during $(t,T]$, the value $d_3$ decreased by at least $x$, and $d_2$ increased by exactly $x$.
We are interested in how large $d_2$ can become.
We consider the following linear program for the values as they are at time $t$. 

\begin{verbatim}
var x>=0;
var d3;
var u1;
var u2;            /* Explanations: */
maximize obj:u2+x; /* upper bound on 1+d2 */
c1:u1+u2+d3<=3;    /* d1+d2+d3<=1 */
c2:u2<=u1+1;       /* a2(t)+phi2(t)<= 2+d1(t) */
c3:d3>=1/4+x;      /* d3>1/4 in interval, so d3(t)>1/4+x */
c4:d3<=5/4;        /* invariant */
c5:d3<=1;          /* assumption, else a1(t)>=1/2 */
\end{verbatim}

The variables $u_i$ in the linear program stand for $a_i+\phi_i=1+d_i$ at time $t$ (either before or after any arrivals at time $t$, it makes no difference). All bounds hold at time $t$; c1 and c4 are invariants which always hold.
We have $a_2^{t^-}\le a_1^t<1/2$. By Lemma \ref{lem:1}, we get $d_3^t<1$ and c5 follows.
By the above, $a_2^T+\phi_2^T$ is upper bounded by $u_2+x$.

If $a_2^{T}\le1$, then $a_2^{T}+\phi_2^{T}\le2$ and we are done.
Suppose $a_2^T>1$ and let $T'\in(t,T)$ be the first time before $T$ at which $a_2^{T'}\le1$ (recall that $a_2^t<1/2$).
We can only have $u_2+x>9/4$, allowing 
$a_2^T+\phi_2^T>9/4$, if $T'-T>1/4$ because $a_2+\phi_2$ increases at speed at most 1 (and only while \PrioGreedy does not run $m_2$). During the entire interval $[T',T]$ 
we had $a_2>1$. Hence $a_1+a_2>1$ during $[T',T]$, so
the last line of the algorithm is not reached. Yet in order for $a_2+\phi_2$ to increase as assumed, \PrioGreedy must have been running $m_3$ for more than $1/4$ time during $[T',T]$.

Therefore we must have had $a_3>5/4$ and $d_3>1/4$ during those times and in particular at time $T$.

Since $d_3$ decreased by at least $x$ overall during $(t,T]$ and a decrease of $d_3$ only happens while \PrioGreedy is working on $m_3$, we must have had $d_3^{t}\ge 1/4+x$, which explains condition c3.

The optimal solution of this linear program is $9/4$. 

\paragraph{Case 2.}
Suppose $a_2^{t}\ge1/2$. 
If we replace c5 by the condition $u_1\ge1/2$, the optimal value is also $9/4$. 
This shows that the competitive ratio is at most $9/4$ as long as $T'\ge t$. 

Suppose $T'<t$. We had $a_2+\phi_2\le 9/4$ until time $t$ inclusive. Thus if we are to have $a_2^T+\phi_2^T>9/4$, then $T>t$, and after time $t$ the value $a_2+\phi_2$ must have increased further. As above it follows that $d_3^T\ge1/4$, implying $d_1^T+d_2^T\le3/4$. To have $a_2^T+\phi_2^T>9/4$ we must have $d_2^T>5/4$ which implies $d_1^T<-1/2$ so $a_1^T<1/2$ But this implies $a_1^t<1/2$ because $a_1$ did not decrease after time $t$, contradicting that $a_2^t\ge1/2$.

It follows in all cases that $d_1+1\le 9/4$ and $d_2+1\le 9/4$ so we have $a_1\le 9/4$ and $a_2\le 9/4$. 
\end{proof}

\noindent The input $I=((0;1,1,0,0),\linebreak[0](1;0,1,1,0),\linebreak[0](2;0,0,1,1),\linebreak[0](3;0,0,0,1),\linebreak[0](5;0,0,1,1),\linebreak[0](6;1,1,1,0),\linebreak[0](9;1,0,1,0),\linebreak[4](10;0,1,1,0),\linebreak[0](11;0,0,1,0))$ shows that the competitive result for \PrioGreedy is tight.

\begin{theorem}\label{thm:NFK3+eLBo}
    In the original model, no (deterministic) online algorithm can be better than $13/6$-compe\-ti\-tive for $K_3+e$.
\end{theorem}
\begin{proof} 
For a contradiction, we consider an algorithm that maintains a competitive ratio of $2+r$ for some value $0<r<1/6$.
We give this algorithm the following input which is repeated many times. In this table, the variable $x$ represents the value of $a_3$ at the beginning of a phase and $*$ indicates that the exact load does not matter. Recall that values of variables at a time $t^-$ ignore any job arrivals at time $t$, whereas values at $t$ include such arrivals. 
When the input starts, $x=0$. We will show that in each phase (that is, each repetition of the input) $x$ increases by at least $1/2-3r>0$. 

\begin{table}[ht]
\centering
\resizebox{0.57\textwidth}{!}{
\begin{tabular}{lccccccccc}
\textbf{Time} & \textbf{Input} & $a_1$ & $a_2$ & $a_3$ & $a_4$ & $z_1$ & $z_2$ & $z_3$ & $z_4$\\ 
\hline 
$0^-$&&0&0&$x$&$*$&0&0&0&$*$\\ 
$0$&$(1,1,1,0)$&1&1&$1+x$&$*$&1&1&1&$*$\\ 
$1^-$&&&&&&1&1&0&$*$\\ 
$1$&&&&&&1&1&0&$*$\\ 
$2^-$&&$*$&$y$&$1+x-2y$&$*$&1&0&0&0\\ 
$2$&$(0,1,1,0)$&$*$&$1+y$&$2+x-2y$&$*$&1&1&1&0\\ 
$3^-$&&$*$&$1+r$&$1+x-r-y$&0&1&1&0&0\\ 
$3$&$(0,0,0,1)$&$*$&$1+r$&$1+x-r-y$&1&1&1&0&1\\ 
$4^-$&&&&&&1&0&0&0\\ 
$4$&$(0,0,0,1)$&&&&&1&0&0&1\\ 
$5^-$&&&&&&0&0&0&0\\ 
$5$&$(0,0,1,1)$&&&&&0&0&1&1\\ 
$6^-/0^-$&&0&0&$1+x-2r-y$&$1+r$&0&0&0&1\\
\end{tabular}
}
\end{table}

At time 2, the load $(1,1,0,0)$ could arrive, as \Opt could be in the state $(0,0,1,0)$ at time $2^-$ (in contrast to what is shown in the table). This could be followed by another unit of load on machine 1 or 2 at time $3$. Since the competitive ratio of \Alg is assumed to be $2+r$, this shows that $a_1^{2^-}\le  \frac12+r$ or  $a_2^{2^-}\le  \frac12+r$. Without loss of generality, we assume for the remainder of the construction that $a_2^{2^-}\le\frac12+r$.

Let the amount of time that \Alg spends processing load on machine 3 in the interval $[0,2]$ be $2y$. Then $a_3^{2^-}=1+x-2y$. \Alg spends (at most) $2(1-y)$ time running machines 1 and 2, so $a_2^{2^-}\ge y$.
We have $a_2^{3^-}\le 1+r$ (else we continue with load $(0,1,0,0)$ as \Opt could have emptied $m_2$ in interval $[2,3]$).
In the interval $[2,3]$, \Alg therefore spends at least $y-r$ time on machine 2, so $a_3^{3^-}\ge 1+x-r-y$. In the interval $[3,6]$, \Alg spends $2-r$ time on machine 4, so $a_3^{6^-}\ge 1+x-2r-y$. This shows that during a phase, the value $a_3$ increases by at least $1-2r-y$.

We have $2y\le 1+x$ since \Alg spends $2y$ time processing load on $m_3$ in the interval $[0,2]$ and only $1+x$ load is available. By the bounds already shown we also have $y\le a_2^{2^-} \le \frac12+r$. Hence the increase in $a_3$ in each phase is at least $1-2r-\min(\frac{1+x}2,\frac12+r)=\max\{1/2-2r-x/2,1/2-3r\}.$

So as long as $r<1/6$ we get a constant increase in the value of $x$ of at least $1/2-3r$ per phase.
This shows that \Alg fails to maintain a competitive ratio below $2+r$.
\end{proof}

\paragraph{Flow Model}
In contrast to the original model, for the path $P_3$ the offline optimum can be matched by an online algorithm  
and hence the path does not yield a lower bound. However, the complete graph $K_3$ is also an induced subgraph of $K_3+e$ and yields a lower bound of $4/3.$

The following algorithm manages to match this bound.

\begin{algorithm}[H]
\label{alg:Fk3+e}
\KwData{$K_3+e$ graph, where $\text{deg}(m_3)=3$ and $\text{deg}(m_4)=1$}
\eIf{ 
(A) $\max(a_1,a_2,a_4)\ge4/3$ and one of these machines with load $\ge4/3$ receives load or $\max(a_1,a_2,a_4)>4/3$}
{run $m_4$ and run \Greedy on $m_1$ and $m_2$;}
{\eIf{ 
(B) $a_3>1/3$ or $a_1=a_2=0$}
    {run $m_3$;}
      {run \Greedy on $m_1,m_2,m_3$ and run $m_4$ if possible}}
  \caption{\PrioGreedyf}
\end{algorithm}

It will be established that the last condition in Case (A) is never satisfied, because this algorithm is $4/3$-competitive.
We will show that our algorithm maintains the following two invariants.
\begin{align}
\label{inv:1}
    d_1+d_2+d_3&\leq 0\\
\label{inv:2}
    d_3+d_4&\leq 1/3 
\end{align}

\begin{lemma}
\label{lem:invfalse}
    Invariant (\ref{inv:1}) can only become false if $a_1=a_2=0$ and $a_4=4/3$ and $m_4$ receives load. As long as (\ref{inv:2}) holds, (\ref{inv:1}) remains true. Invariant (\ref{inv:2}) can only become false if $\max(a_1,a_2)\ge4/3$, $\min(a_1,a_2)<1/3$, $a_4=0$ and $m_4$ does not receive load. 
\end{lemma}
\begin{proof}
    Both invariants hold at the start.
    If (A) is not satisfied, at least one machine of $m_1,m_2$ or $m_3$ is run and hence $d_1+d_2+d_3$ does not increase, or $a_1=a_2=a_3=0$ which immediately implies (\ref{inv:1}). We see that (\ref{inv:1}) can only become false while (A) is satisfied and only if \Greedy runs none of the machines $m_1,m_2,m_3$. This implies $a_1=a_2=0$, $a_4\ge4/3$ and $m_4$ receives load. By (\ref{inv:2}), $d_3\le0$ at this time, so (\ref{inv:1}) continues to hold as long as (\ref{inv:2}) holds.
    
    We now consider (\ref{inv:2}).
    If (A) is not satisfied, then if $a_4>0$, at least one machine $m_3$ or $m_4$ is run and hence $d_3+d_4$ does not increase, and while $a_4=0$ the value $d_3+d_4$ can only increase while $a_3\le1/3$, implying (\ref{inv:2}). So (\ref{inv:2}) can only become false while (A) is satisfied. If additionally $a_4>0$ or $m_4$ receives load, then $m_4$ is run and $d_3+d_4$ does not increase. Hence for (\ref{inv:2}) to become false we must have $a_4=0$ and $m_4$ does not receive load. To satisfy (A) with $a_4=0$ we must have $\max\{a_1,a_2\}\ge4/3$. Wlog let $a_2=\max\{a_1,a_2\}$.
    Once invariant (\ref{inv:2}) becomes false, we have $d_3>1/3,d_2\ge 1/3$ and hence $d_1<-2/3$ by invariant (\ref{inv:1}) and it follows $a_1<1/3$. Because $a_1$ does not decrease as we run $m_2$ (which is needed for (\ref{inv:2}) to become false), this was also true just before invariant (\ref{inv:2}) became false.
\end{proof}
\begin{lemma}
\label{lem:notbelow}
During execution of the algorithm the following holds.
\begin{thmlist}
    \item\label{lem:notbelow1} If $a_3\ge1/3$ and $\max(a_1,a_2)>1/3$, then  $a_3\ge1/3$ is maintained as long as $\max(a_1,a_2)>1/3$. 
    \item\label{lem:notbelow2} If $\min(a_1,a_2)\ge1/3$,  then this is maintained as long as $\max(a_1,a_2)\ge1/3$.
    \item\label{lem:notbelow3} While $a_3>1/3$ and $\max(a_1,a_2)<4/3$, if $a_4\ge0$, then $a_4\ge0$ continues to hold. 
    \item\label{lem:notbelow4} While $a_3>1/3$ and $a_4<4/3$, if $\max(a_1,a_2)\ge4/3$,  then $\max(a_1,a_2)\ge4/3$ continues to hold. 
\end{thmlist}
\end{lemma}
\begin{myproof}
\begin{enumerate}[label=\textit{\arabic*}.]
    \item Machine $m_3$ is not run if (A) holds. If (B) holds, we have $a_3>1/3$. Finally, if (A) and (B) do not hold $a_3$ cannot drop below $1/3$ as long as $\max(a_1,a_2)>1/3$, since \Greedy would run one of those machines.
    \item \Greedy is used to decide which machine $m_1$ or $m_2$ is run. As long as $\min(a_1,a_2)<\max(a_1,a_2)$ the machine with load $\max(a_1,a_2)$ is run. This changes once  $\min(a_1,a_2)=\max(a_1,a_2),$ at which time both machines are run.
    \item Condition (B) holds and hence if (A) does not hold $m_4$ is not run. If (A) holds, then by $\max(a_1,a_2)<4/3$ we have $a_4\ge4/3$ and if $a_4=4/3$, $m_4$ must receive load. It follows that its load does not drop below $4/3.$
    \item  Condition (B) holds and hence if (A) does not hold machines $m_1$ and $m_2$ are not run. If (A) holds, then by $a_4<4/3$, there must be at least one machine $m_1$ or $m_2$ with load at least $4/3$. If $\max\{a_1,a_2\}=4/3$, at least one of these machines receives load. It follows that its load does not drop below $4/3.$\hfill \qed
\end{enumerate}
\end{myproof}

\begin{definition}
    A \emph{critical interval} is a maximal open interval
in which we have $\max(a_1,a_2)>1/3$ and $\min(a_1,a_2)<1/3$ throughout.
\end{definition}
A critical interval is always bounded because the input ends at some point, so the buffers eventually become empty. 

\begin{lemma}
    \label{lem:d3t1}
    Let $(t_1,t_2)$ be a critical interval.
    We have $\max(a_1,a_2)\le4/3$ throughout the critical interval, and 
    $a_2^{t_1}=1/3$.
    If (\ref{inv:2}) holds at time $t_1$, then $d_3^{t_1}\le1/3$.
    If additionally $a_3^{t_1}>1/3$, then $d_3^{t_1}\le0$ and $a_4^{t_1}\ge4/3$. 
\end{lemma}
\begin{proof}
    Without loss of generality, let $a_2=\max(a_1,a_2)$ in this interval.  Then $a_1<1/3<a_2$ in the entire interval. We have $a_1^{t_1}<1/3$ (by Lemma \ref{lem:notbelow}.2 and the previous line) and hence $a_2^{t_1}=1/3$ by the maximality of a critical interval.
    By condition (A), $a_2\le4/3$ is maintained during this interval: $m_2$ will always be running if $a_2\ge4/3$ and $m_2$ is receiving load, since $a_1<1/3$.
    
    If $a_3^{t_1}\le1/3$ there is nothing more to show.
    Suppose $a_3^{t_1}>1/3$. Let $t_0$ be the last time before $t_1$ such that $a_3^{t_0}=1/3$.
    In this case condition (B) was satisfied in the interval $(t_0,t_1)$, but still $a_3$ increased since then. 

Suppose for a contradiction that $a_4<4/3$ in the entire interval $(t_0,t_1)$.
Then $a_1$ does not decrease below $1/3$ in $(t_0,t_1)$ since $a_3>1/3$ and because (A) can only be true in this case if $a_2=4/3$, and then $a_1$ also does not decrease below $1/3$. Because $a_1^{t_1}<1/3$ we have $a_1<1/3$ in the entire interval $(t_0,t_2)$.
We have $a_3>1/3$ during $(t_0,t_1)$ and $a_3$ increased during some time in $(t_0,t_1)$, meaning that $a_2\ge4/3$ for some time in $(t_0,t_1)$, since $a_1<1/3$ and by assumption $a_4<4/3$ throughout. 
By Lemma \ref{lem:notbelow}.4 it follows that $a_2^{t_1}=4/3$, contradicting the definition of $t_1$. Hence the assumption $a_4<4/3$ was wrong and 
there was a time $t'\in (t_0,t_1)$ at which $a_4=4/3$. 
From this point on $a_4$ remains equal to $4/3$ by Lemma \ref{lem:notbelow}.3 unless $\max(a_1,a_2)=4/3$ at some point. In such a case, $a_4$ can only drop below $4/3$ if $a_1$ or $a_2$ remains $4/3$, and then that load cannot drop below $4/3$ anymore (unless $a_4$ becomes equal to $4/3$ again) by Lemma \ref{lem:notbelow}.4. However, since $a_1^{t_1}<1/3$ and $a_2^{t_1}=1/3$, we must have $a_4^{t_1}=4/3$.

Therefore $d_4^{t_1}\ge1/3$ and by the invariant $d_3+d_4\le1/3$, we have $d_3^{t_1}\le0$. 
\end{proof}

\begin{observation} 
    In a critical interval, whenever $a_4$  increases, $d_3$ does not increase because \PrioGreedyf is running $m_3$.
If we reach $a_4=4/3$ again, then $d_3\le0$ as long as invariant \ref{inv:2} holds. 
\end{observation}

\begin{lemma}
\label{lem:FK3+eUB}
Algorithm \PrioGreedyf maintains (\ref{inv:1}) and (\ref{inv:2}).
\end{lemma}

\begin{proof}
By Lemma \ref{lem:invfalse}, we only need to show that (\ref{inv:2}) is maintained. 
Let $(t_1,t_2)$ be a critical interval in which (\ref{inv:2}) holds at time $t_1$ and in which
there is a time at which $\max(a_1,a_2)=4/3$, $\min(a_1,a_2)<1/3$, $a_4=0$ and $a_3\ge1/3$, and let $T$ be the {first} such time in that interval. By Lemma \ref{lem:invfalse} the invariants are maintained in the interval $(t_1,T]$.
By Lemma \ref{lem:d3t1}, $a_2\le4/3$ throughout this interval.

By lemma \ref{lem:d3t1}, we have $a_2^{t_1}=1/3$, hence in the interval $[t_1,T]$, the value $a_2$ increases by 1. The value $d_2+d_3$ does not increase because the algorithm either works on $m_2$ or $m_3$ as $a_2>1/3$; the algorithm never works on $m_1$ because $a_1<1/3< a_2$.
We have $a_2^T-a_2^{t_1}=1$. There are two cases.

\paragraph*{Case 1: $a_3^{t_1}\le1/3$} In this case $a_3^T\ge a_3^{t_1}$ and we get
\[d_2^T+d_3^T\le d_2^{t_1}+d_3^{t_1} \Rightarrow
z_2^T+z_3^T\ge 1+a_3^T-a_3^{t_1}\ge 1+a_3^T-1/3.\] 

\paragraph*{Case 2: $a_3^{t_1}>1/3$}
After time $t_1$, $a_3$ decreases by exactly $a_3^{t_1}-a_3^T$ (this amount may be negative) whereas $a_2$ increases by exactly 1. The overall increase in the value $a_2+a_3$ is therefore $1+a_3^T-a_3^{t_1}$. 
Since $d_2+d_3$ does not increase in the interval $[t_1,T]$, we have that $z_2+z_3$ increases by at least the same amount. With Lemma \ref{lem:d3t1} we also have $d_3^{t_1}\le0$, which leads to \[z_2^T+z_3^T\ge a_3^{t_1}+1+a_3^T-a_3^{t_1}=1+a_3^T.\] 

By Lemma \ref{lem:notbelow}.1, $a_3\ge1/3$ in $[T,t_2]$ because $a_2>1/3$ in $[T,t_2)$.
Consider the machines $m_2$ and $m_3$ at a time $t\in[T,t_2)$. 
Suppose $a_4$ remains below $4/3$ in $[T,t].$
In any part of a critical interval in which $a_4\le4/3$, machine $m_2$ only runs if it receives load and $a_2=4/3$, or if $a_3\le 1/3$.
Let $t_{23}^+$ be the total length of time that both machines received load during $[T,t]$, and let $t_{23}^-$ be the total length of time that both machines  did not receive load during $[T,t]$.
During the remaining time (if any), exactly one of these machines receives load and $z_2+z_3$ does not decrease. If $m_2$ is the machine receiving load, $a_3$ does not increase; if $m_3$ is the machine receiving load, then \Alg runs $m_3$ using condition (B) because $a_3\ge1/3$ in $[T,t_2]$ (and because it does not run $m_2$ in this case) and again $a_3$ does not increase.

Since $z_2^t\le1$, we have 
\begin{equation}
\label{eq:z3t}
    z_3^t\ge t_{23}^+-t_{23}^-+
    \begin{cases}
        a_3^T-1/3 & \text{ if } a_3^{t_1}\le1/3\\
        a_3^T & \text{ if }a_3^{t_1}>1/3
    \end{cases}
\end{equation}
whereas 
\begin{equation}
    \label{eq:a3t}
    a_3^t\le a_3^T+t_{23}^+-t_{23}^-    
\end{equation}
as long as the right hand side of (\ref{eq:a3t}) is more than $1/3$
because \Alg works on $m_2$ when both $m_2$ and $m_3$ receive load and $a_2=4/3$, as well as when $a_2>4/3$, 
and on $m_3$ otherwise by condition (B), since $a_4<4/3$. 

By the definitions of $t_{23}^+$ and $t_{23}^-$ and because $a_2^T=4/3$, for $t\in[T,t_2)$ we also have
\begin{equation}
    \label{eq:a2a3}
    a_2^t+a_3^t = 4/3+a_3^T+t_{23}^+-t_{23}^-.
\end{equation}

If $a_3\le1/3$ and the right hand side of (\ref{eq:a3t}) is more than $1/3$, clearly (\ref{eq:a3t}) also holds.
During intervals in which the right hand side of (\ref{eq:a3t}) is at most $1/3$, $a_3$ remains equal to $1/3$ because $a_2>1/3$, so $d_3\le 1/3$. The value $a_3$ can only become larger again once $a_2=4/3$ again, which happens when $a_3^t=a_3^T+t_{23}^+-t_{23}^-$ by (\ref{eq:a2a3}). We conclude that (\ref{eq:a3t}) holds throughout.

Since both (\ref{eq:z3t}) and (\ref{eq:a3t}) hold, we have $d_3^t= a_3^t-z_3^t\le 0$ if $a_3^{t_1} > 1/3$, else $d_3^t\le 1/3$.

If $a_4\ge4/3$ at any point in $[T,t]$, then after this $d_3+d_4$ does not increase until possibly $a_4=0$ again, and this also requires that $\max(a_1,a_2)=4/3$ again. So in this case both invariants hold at least until the next time that 
$\max(a_1,a_2)=4/3, \min(a_1,a_2)<1/3, a_4=0$ and $a_3\ge1/3$ by the above. If this happens before time $t_2$, then we define this to be our new time $T$ and continue as above. In this case the value $t_1$ remains the same (so if we were previously in Case 2, we now still have $a_4^{t_1}\ge4/3$ and $d_3^{t_1}\le0$). Otherwise we are done with the current interval. 

This proof shows that (\ref{inv:2}) is maintained during any critical interval in which it holds at the start of that interval, so by Lemma \ref{lem:invfalse} both invariants hold throughout. 
\end{proof}

\begin{theorem}
\label{thm:FK3+eUB}
Algorithm \PrioGreedyf is $4/3$-competitive on $K_3+e.$ 
\end{theorem}

\begin{proof}[Proof of \cref{thm:FK3+eUB}]
The value $a_4$ never increases above $4/3$ by condition (A). 
Suppose $\max(a_1,a_2)>4/3$ at some time. At this time, it must be that $a_1=a_2$ and both of these machines are receiving load by condition (A).
Then $d_1+d_2>2/3$. By Lemma \ref{lem:FK3+eUB} we have maintained $d_1+d_2+d_3\le0$ up until this point. Hence $d_3<-2/3$, meaning $a_3<1/3$. In this case $m_3$ has not been running since $\max(a_1,a_2)=a_3$ for the last time, because since that time $a_3$ did not decrease since (B) did not hold.
At that point $d_1+d_2\le2\max(a_1,a_2)=2a_3<2/3$, and $d_1+d_2$ did not increase from then on, which contradicts $d_1+d_2>2/3$. 
This shows that $\max(a_1,a_2)\le4/3$ as long as $d_1+d_2+d_3\le0$.

The only thing left to show is $a_3\le4/3$. This can only become false while (A) is satisfied, as otherwise $a_3$ does not increase while $a_3>1/3$.
If $a_4\ge4/3$, then $d_4\ge1/3$ so $d_3\le0$ by (\ref{inv:2}) and hence $a_3\le1$. Consider an open interval in which $\max(a_1,a_2)>1/3$. Suppose $\min(a_1,a_2)=1/3$ at the start of this interval. Then by Lemma \ref{lem:notbelow}.2 we have $\min(d_1,d_2)\ge-2/3$ in this entire open interval, and whenever $\max(d_1,d_2)\ge1/3$ this implies $d_3\le0-\max(d_1,d_2)-\min(d_1,d_2)\le1/3$ by (\ref{inv:1}).  

Otherwise, suppose $\min(a_1,a_2)<1/3$ at the start of this interval. 
This is an interval $(t_1,t_2)$ that we have considered before.
The proof of Lemma \ref{lem:FK3+eUB} shows that $d_3\le1/3$ after time $T$ as long as $a_4<4/3$. 
Starting from time $t_1$, there can be multiple intervals (maybe of 0 length) in which $a_4=4/3$. In such interval, we have $d_3\le0$ by (\ref{inv:2}), and between two intervals $d_3$ cannot increase above $1/3$ until we reach time $T$, because $\max(a_1,a_2,a_4)<4/3$ 
there so we are not in situation (A), 
so \PrioGreedyf runs $m_3$ whenever $a_3>1/3$.
This shows that $d_3\le0$ in $(t_1,T]$. If there is no time $T$, it shows that $d_3\le0$ in $(t_1,t_2)$. If we reach $a_4=4/3$ after time $T$, then just as above $d_3$ cannot increase above $1/3$ until we possibly reach the next time that we call $T$ (see the previous proof).
\end{proof}

\begin{theorem}\label{thm:k3+eLB}
    No randomized online algorithm can be better than $4/3$-competitive for $K_3+e.$
\end{theorem}
\begin{proof}
$K_3$ is an induced subgraph of $K_3+e$. The statement follows by the lower bound for complete graphs (Theorem \ref{thm:FKnLB}).
\end{proof}

\section{Conclusions}
Our results on complete partite graphs show that the number of vertices $n$ is certainly not the sole factor determining the competitive ratio. These results also show that the flow model is not significantly easier than the original model even though load appears more slowly; for complete $k$-partite graphs, the results are asymptotically the same. Generally, it seems that the density of a graph is an important factor in what competitive ratios can be achieved. Symmetry helps to analyze graphs, as we see for the graphs on four vertices. The special structure of each individual graph has a significant effect on the competitive ratio, and this complicates attempts to generalize these results further.

\paragraph{Acknowledgments}
We gratefully acknowledge Marek Chrobak and Ji{\v{r}}{\'i} Sgall for valuable discussions.


\end{document}